\documentclass[twoside,12pt]{article}
\usepackage{amsmath}
\usepackage{amsthm}
\usepackage{amssymb}
\usepackage{amscd}
\usepackage[british]{babel}
\usepackage{graphicx}
\usepackage{pb-diagram}
\usepackage{times}
\usepackage{url}
\usepackage{psfrag}
\usepackage{color}

\theoremstyle{plain}
\newtheorem{definition}{Definition}[section]

\newtheorem{lemma}[definition]{Lemma}

\newtheorem{proposition}[definition]{Proposition}

\newtheorem{remark}[definition]{Remark}

\newcommand{\boxend}{\flushright{$\Box$}}

\newcommand{\br}{{\mathbb R}}

\newcommand{\bt}{{\mathbb T}}

\newcommand{\al}{\alpha}
\renewcommand{\epsilon}{\varepsilon}
\newcommand{\ep}{\varepsilon}

\renewcommand{\phi}{\varphi}

\newcommand{\cala}{{\mathcal A}}
\newcommand{\calb}{{\mathcal B}}
\newcommand{\calc}{{\mathcal C}}

\newcommand{\calh}{{\mathcal H}}
\newcommand{\calo}{{\mathcal O}}

\newcommand{\rmi}{{\rm i}}
\newcommand{\rme}{{\rm e}}
\newcommand{\id}{\mbox{\rm Id}}

\renewcommand{\tilde}{\widetilde}

\title{A dynamical systems approach to Bohmian trajectories}

\author{F. Borondo\thanks{f.borondo@uam.es}\\
  Departamento de Qu\'{\i}mica, and \\
  Instituto Mixto CSIC--UAM--UC3M--UCM,\\
  Universidad Aut\'onoma de Madrid,\\
  Cantoblanco--28049 Madrid, Spain  \and
  A. Luque\thanks{alejandro.luque@upc.edu} and
  J. Villanueva\thanks{jordi.villanueva@upc.edu}\\
  Departament de Matem\`atica Aplicada I,\\
  Universitat Polit\`ecnica de Catalunya,\\
  08028 Barcelona, Spain \and
  D. A. Wisniacki\thanks{wisniacki@df.uba.ar}\\
  Departamento de F\'{\i}sica ``J. J. Giambiagi'', \\
  FCEN, UBA, Pabell\'on 1, Ciudad Universitaria, \\
  1428 Buenos Aires, Argentina}

\begin{document}

\maketitle

\thispagestyle{empty}
\newpage
\begin{abstract}
Vortices are known to play a key role in the dynamics of the
quantum trajectories defined within the framework of the de
Broglie-Bohm formalism of quantum mechanics.
It has been rigourously proved that the motion of a vortex
in the associated velocity field can induce chaos in these
trajectories,
and numerical studies have explored the rich variety of
behaviors that due to their influence can be observed.
In this paper, we go one step further and show how the theory
of dynamical systems can be used to construct a general and
systematic classification of such dynamical behaviors.
This should contribute to establish some firm grounds on
which the studies on the intrinsic stochasticity of Bohm's
quantum trajectories can be based.
An application to the two dimensional isotropic harmonic
oscillator is presented as an illustration.
\end{abstract}

\markboth{Bohmian trajectories from a Dynamical Systems point of view}
{F. Borondo, A. Luque, J. Villanueva and D.A. Wisniacki}

\newpage

\section{Introduction}\label{sec:intro}

Some interpretational difficulties~\cite{Feynman} with the standard
version~\cite{vonNeuman} led David Bohm to develop in the 1950's
\cite{Bohm} an alternative formulation of quantum mechanics.
Despite initial criticisms, this theory has recently received
much attention \cite{Holland,Wyattbook},
having experimented in the past few years an important revitalization,
supported by a new computationally oriented point of view.
In this way, many interesting practical applications,
including the analysis of the tunnelling mechanism
\cite{Makri,LiuMakri,LopreoreW99},
scattering processes \cite{Gonzalez,Beswick,SanzBM04},
or the classical--quantum correspondence \cite{PrezhdoB01,SanzBM01},
just to name a few, have been revisited using this novel point of view.
Also, the chaotic properties of these trajectories
\cite{Frisk97,WisniackiP05,WisniackiPB05,WisniackiPB07,Efthymiopoulos},
or more fundamental issues,
such as the extension to quantum field theory \cite{DurrGTZ04},
or the dynamical origin of Born's probability rule \cite{ValentiniW05}
(one of the most fundamental cornerstones of the quantum theory)
have been addressed within this framework .

Most interesting in Bohmian mechanics is the fact that this
theory is based on quantum trajectories, ``piloted'' by the
de Broglie's wave which creates a (quantum) potential term
additional to the physical one derived from the actual forces
existing in the system \cite{Bohm}.
This term brings into the theory interpretative capabilities
in terms of intuitive concepts and ideas, which are naturally deduced
due to fact that quantum trajectories provide causal connections
between physical events well defined in configuration and time.
Once this ideas have been established as the basis of many numerical studies,
it becomes, in our opinion, of great importance to provide
firm dynamical grounds that can support the arguments based on
quantum trajectories.
For example, it has been recently discussed that the chaotic properties
of quantum trajectories are critical for a deep understanding of
Born's probability quantum postulate, considering it as an emergent
property \cite{ValentiniW05}.
Unfortunately very little progress,
i.e. rigorous formally proved mathematical results,
has taken place along this line due to the lack of a solid theory
that can foster this possibility.
Moreover, there are cases in the literature clearly demonstrating
the dangers of not proceeding in this way.
One example can be found in Ref.~\cite{polavieja}, where a chaotic
character was ascribed to quantum trajectories for the quartic potential,
supporting the argument solely on the fact that numerically computed
neighboring pairs separate exponentially.
This analysis was clearly done in a way in which the relative importance
of the quantum effects could not be gauged.
Something even worse happened with the results reported in \cite{brasileno},
that were subsequently proved to be wrong in a careful
analysis of the trajectories \cite{brasileno2}.

Recently, some of the authors have made
in Refs.~\cite{WisniackiP05,WisniackiPB05,WisniackiPB07}
what we consider a relevant advance along the line
proposed in this paper,
by considering the relationship between the eventual
chaotic nature of quantum trajectories and the vortices existing in
the associated velocity field which is given by the quantum potential,
a possibility that had been pointed out previously by
Frisk \cite{Frisk97}.
Vortices has always attracted the interest of scientists from many
different fields.
They are associated to singularities at which certain mathematical
properties become infinity or abruptly change,
and play a central role to explain many interesting phenomena both
in classical and quantum physics \cite{Berry00}.
In these papers it was shown that
quantum trajectories are, in general, intrinsically chaotic,
being the motion of the velocity field vortices a sufficient mechanism
to induce this complexity \cite{WisniackiP05}.
In this way, the presence of a single moving vortex, in an otherwise
classically integrable system, is enough to make quantum trajectories chaotic.
When two or few vortices exist, the interaction among them
may end up in the annihilation or creation of them in pairs with
opposite vorticities. These phenomenon makes that the size of the
regular regions in phase space grows as vortices disappear \cite{WisniackiPB07}.
Finally, it has been shown that when a great number of vortices are present
the previous conclusions also hold, and they statistically combine in
such a way that they can be related with a suitably defined Lyapunov
exponent, as a global numerical indicator of chaos in the
quantum trajectories \cite{WisniackiPB05}.
Summarizing, this makes of chaos the general dynamical scenario for
quantum trajectories, and this is due to the existence and motion
of the vortices of the associated velocity field.

In this paper, we extend and rigorously justify the numerical results
in~\cite{WisniackiPB05,WisniackiPB07,WisniackiP05}
concerning the behavior of quantum trajectories and its structure
by presenting the general analysis of a particular problem of general interest,
namely a two--dimensional harmonic oscillator,
where chaos does not arise from classical reasons.
In this way, we provide a systematic classification of all possible
dynamical behaviors of the existing quantum trajectories,
based on the application of dynamical systems theory
\cite{dynsys}.
This classification provides a complete ``road--map'' which makes possible a
deep understanding, put on firm grounds, of the dynamical structure for
the problem being addressed.

\section{Bohmian mechanics and quantum trajectories}

The Bohmian mechanics formalism of quantum trajectories starts from
the suggestion made by Madelung of writing the wave function in polar form
\begin{equation}
  \psi(r,t)=R(r,t) \rme^{\rmi S(r,t)}, \nonumber
\end{equation}
where $R^2=\overline{\psi} \psi$ and
$S = (\ln \psi-\ln \overline{\psi})/(2\rmi)$
are two real functions of position and time.
For convenience, we set $\hbar=1$ throughout the paper, and consider a
particle of unit mass.
Substitution of this expression into the time-dependent Schr\"odinger
equation allows to recast the quantum theory into a ``hydrodynamical'' form
\cite{Wyattbook}, which is governed by,
\begin{align}
\frac{\partial R^2}{\partial t} & = - \nabla \cdot \bigg( R^2\nabla S\bigg),
  \label{eq:bohmPDE1}\\
\frac{\partial S}{\partial t} & = -\frac{(\nabla S)^2}{2} - V -
\frac{1}{2}\frac{\nabla^2 R}{R},
  \label{eq:bohmPDE2}
\end{align}
which are the continuity and the ``quantum'' Hamilton-Jacobi equations,
respectively.
The qualifying term in the last expression is customarily included since
this equation contains an extra non-local contribution
(determined by the quantum state),
$Q=\frac{1}{2}\nabla^2 R/R$, called the ``quantum'' potential.
Together with $V$, this additional term determines the total forces
acting on the system, and it is responsible for the so-called quantum
effects in the dynamics of the system.

Similarly to what happens in the standard Hamilton-Jacobi theory,
Eqs.~\eqref{eq:bohmPDE1} and~\eqref{eq:bohmPDE2} allow to define,
for spinless particles, quantum trajectories by integration of the
differential equations system: $\ddot{r}=-\nabla V(r)-\nabla Q(r)$.
Alternatively, one can consider the velocity vector field
\begin{equation}\label{eq:Vfield}
  X_\psi  = \nabla S = \frac{\rmi}{2}\frac{\psi \nabla
  \overline{\psi}-\overline{\psi} \nabla \psi}{|\psi|^2}.
\end{equation}
Notice that, in general, this Bohmian vector field is not Hamiltonian,
but it may nevertheless have some interesting properties.
In particular, for the example considered in this paper it will be
shown that it is time-reversible, this symmetry allowing the study
of its dynamics in a systematic way.

Let us recall that a system, $\dot{r}=X(r,t)$, is time-reversible
if there exists an involution, $r=\Theta(s)$, that is a change of
variables satisfying $\Theta^2 = \id$ and $\Theta \neq \id$,
such that the new system results in
$\dot{s}=D\Theta^{-1}(s) X(\Theta(s),t)=-X(s,t)$.
One of the dynamical consequences of reversibility is that if
$r(t)$ is a solution, then so it is $\Theta(r(-t))$.
This fact introduces symmetries in the system giving rise to
relevant dynamical constraints.
For example, let us assume that $\Theta(x,y)=(x,-y)$ is a
time-reversible symmetry (see Fig.~\ref{fig:reversible}).
Then any solution $r(t)=(x(t),y(t))$ defines another solution
given by $(x(-t),-y(-t))$.
Let us remark that this fact constraints the system dynamics since if,
for example, $r(t)$ crosses the symmetry axis ($y=0$ is invariant under
$\Theta$) then the two solutions must coincide.

\begin{figure}
\centering
\psfrag{x}{{\footnotesize $r_1(t)$}}
\psfrag{tx}{{\footnotesize$\Theta(r_1(-t))$}}
\psfrag{y}{{\footnotesize$r_2(t)$}}
\psfrag{Ty}{{\footnotesize$\Theta(r_2(-t))$}}
\psfrag{eje}{{\footnotesize$y=0$}}
\includegraphics[width=7cm,height=6cm]{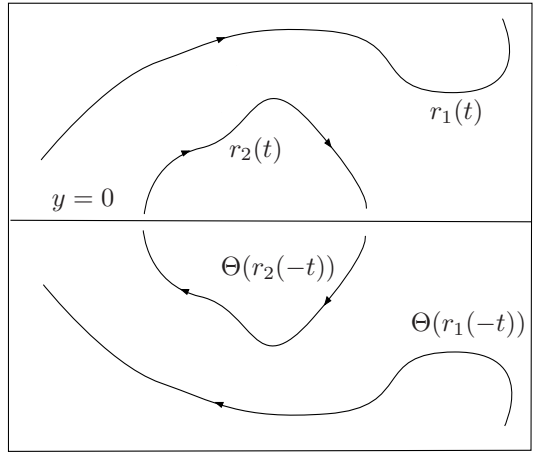}
\caption{{\footnotesize Illustration of the dynamical consequences of a
time--reversible symmetry.}}
  \label{fig:reversible}
\end{figure}

We conclude this section by stressing that time-reversible systems
generated a lot of interest during the 80's due to the fact that
they exhibit most of the properties of Hamiltonian systems
(see \cite{Arnold84,Broer96,Sevryuk07}).
In particular, this type of systems can have quasi-periodic tori
which are invariant under both the flow and the involution $\Theta$.
That is, KAM theory fully applies in this context.
Furthermore, some interesting results concerning the splitting
of separatrices have been developed successfully for time-reversible
systems~\cite{Lazaro03}, providing powerful tools for the study of
homoclinic and heteroclinic chaos.

\section{Model and canonical form}\label{sec:model}

The system that we choose to study is the two dimensional isotropic
harmonic oscillator.
Without loss of generality, the corresponding Hamiltonian operator for
$r=(x,y)$ can be written in the form
$$
\hat H(x,y) = -\frac{1}{2} \bigg(\frac{\partial^2}{\partial x^2} +
\frac{\partial^2}{\partial y^2}\bigg) + \frac{1}{2} (x^2+y^2).
$$
In this paper, we consider the particular combination of eigenstates:
$\phi_{0,0}=1/\sqrt{\pi}$ with energy $1$, and
$\phi_{1,0} = 2x/\sqrt{2\pi}$, $\phi_{0,1} = 2y/\sqrt{2\pi}$ with energy $2$.
It can be immediately checked that the time evolution of the resulting
wave function is given by
\begin{equation}\label{eq:wave:ini}
\psi = \bigg( \frac{\cala \rme^{-\rmi t}}{\sqrt{\pi}}  +
  \frac{2 x \calb \rme^{-2\rmi t}}{\sqrt{2\pi}} +
  \frac{ 2 y \calc \rme^{-2 \rmi t}}{\sqrt{2\pi}}  \bigg)
  \rme^{-\frac{1}{2}(x^2+y^2)},
\end{equation}
where\footnote{The choice of notation for the real and imaginary parts
of $\calc$ may look arbitrary at this point, but it makes simpler
the notation for the canonical form introduced in the next section.}
$\cala=A +\rmi D$, $\calb = B + \rmi E$ and $\calc = F+\rmi C$,
subject to the usual normalization condition $|\cala|^2 + |\calb|^2 +|\calc|^2 = 1$.
In addition, we further assume the condition $BC \neq EF$ in order to
ensure the existence of a unique vortex in the velocity field at any time.
Accordingly, the quantum trajectories associated to~\eqref{eq:wave:ini}
are solutions of the system of differential equations:
\begin{eqnarray}\label{eq:case3}
\dot{x} & = & \frac{-2(BC-EF)y-\sqrt{2}(BD-AE) \cos t -
\sqrt{2}(AB+DE) \sin
t}{V(x,y,t)}, \label{eq:case3:1} \\
\dot{y} & = & \frac{2(BC-EF)x+\sqrt{2}(AC-DF)\cos t -
\sqrt{2}(DC+AF) \sin t}{V(x,y,t)}, \label{eq:case3:2}
\end{eqnarray}
where
\begin{eqnarray*}
V(x,y,t) & = & 2 (B^2+E^2)x^2 + 2 (C^2+F^2) y^2 + 4 (BF+EC) xy + D^2+A^2\\
         &   & +2\sqrt2((AB+DE)\cos t + (AE-DB) \sin t)x\\
         &   & +2\sqrt2((DC+AF)\cos t + (AC-DF) \sin t)y.
\end{eqnarray*}
To integrate this equation a 7/8--th order Runge-Kutta-Fehlberg
method has been used.
Moreover, since the vector field is periodic, the dynamics can
be well monitored by using stroboscopic sections.
In particular, we plot the solution $(x(t),y(t))$
at times $t=2\pi n$ for $n=1,2,\ldots,10^4$ and for several initial
conditions.

\begin{figure}
\centering
\includegraphics[width=6cm,height=6cm,angle=-90]{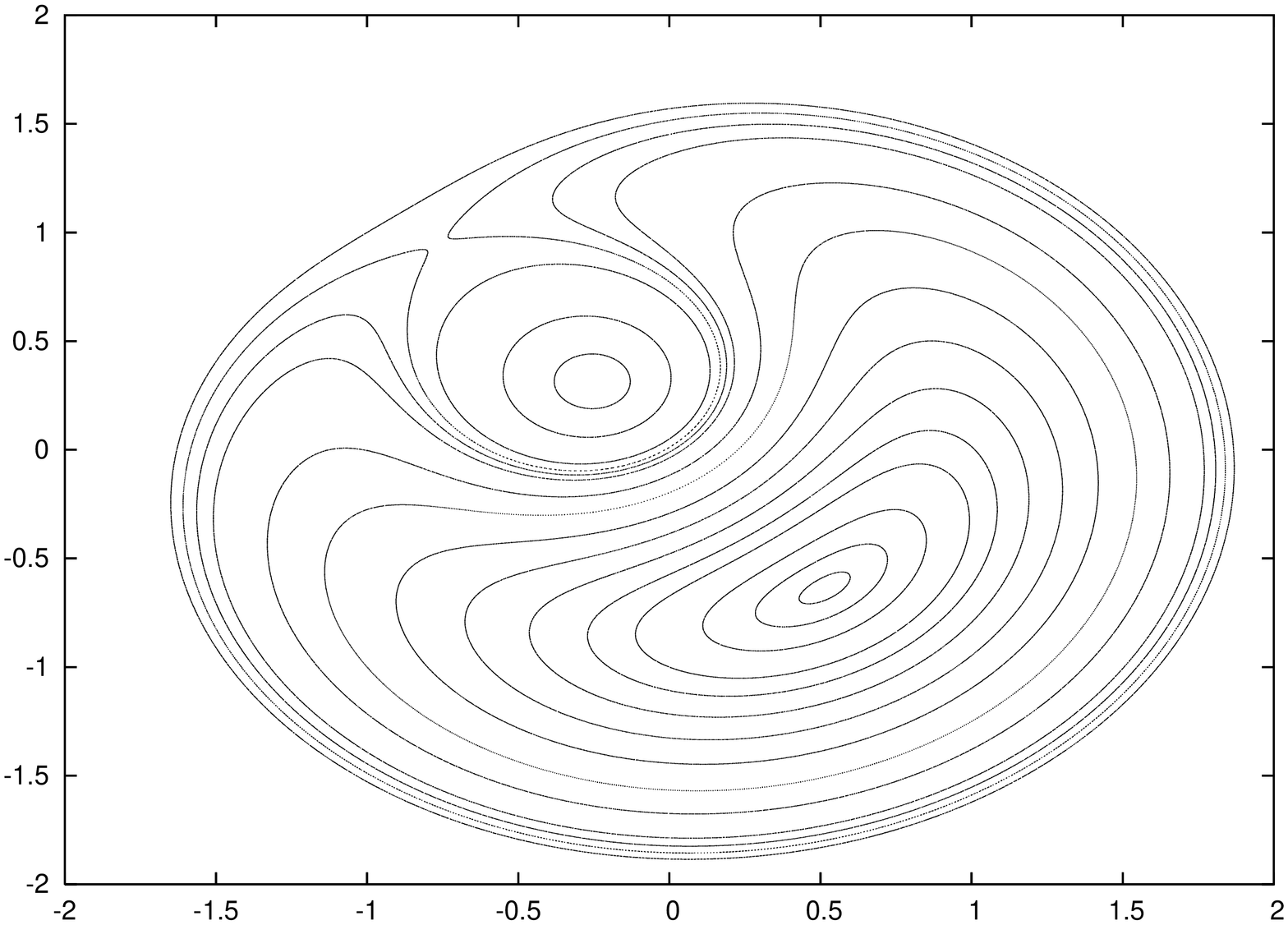}
\includegraphics[width=6cm,height=6cm,angle=-90]{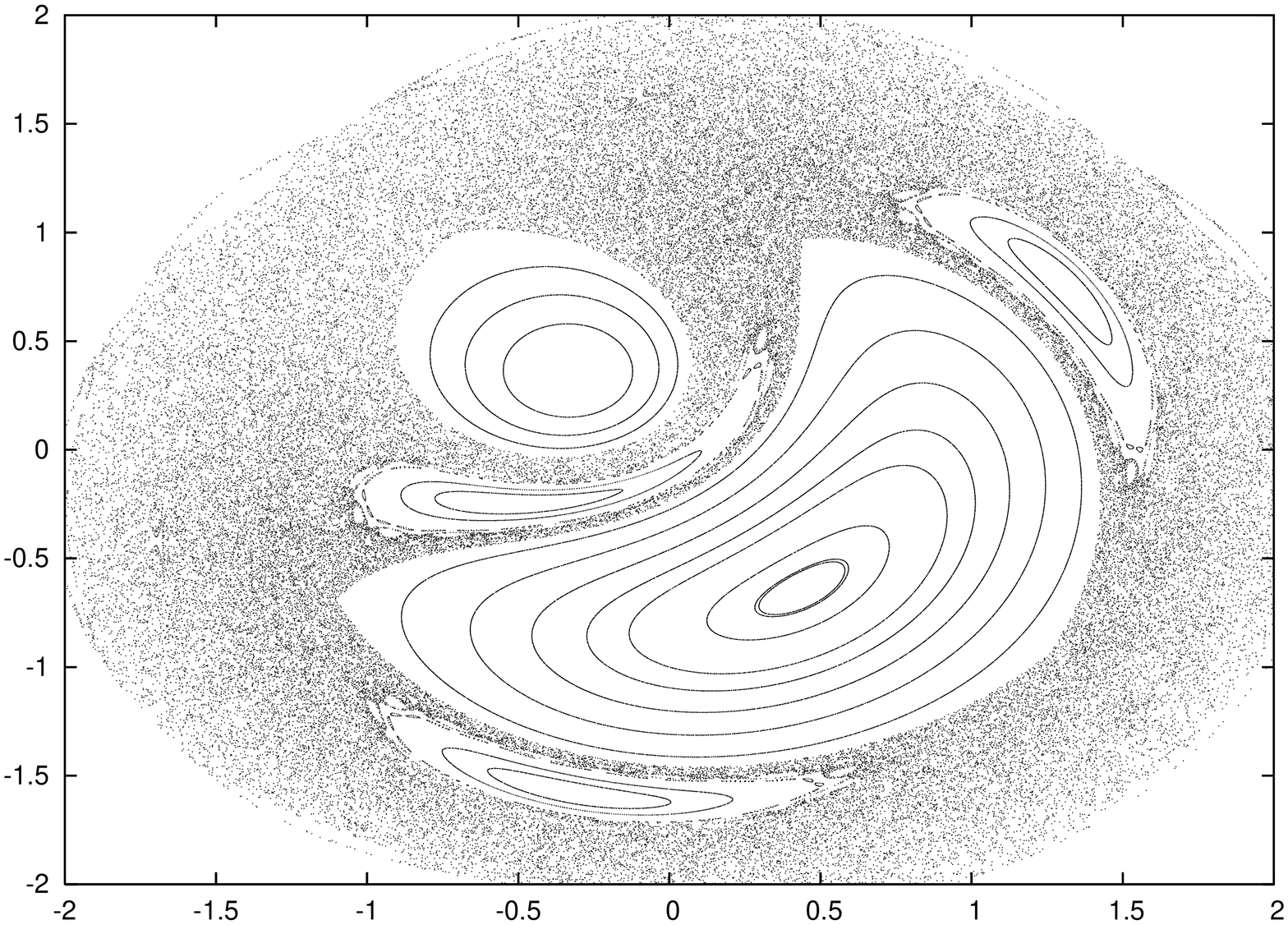}
\caption{{\footnotesize Stroboscopic $2\pi$-periodic sections
  for the quantum trajectories generated by Eqs.~\eqref{eq:case3}
  and \eqref{eq:case3:2} for different values of the normalized
  constants. \protect\\
  Left plot: $A = 0.37$, $D =-0.02$, $B = C = 0.44$ and $E = -F$.
    \protect\\
  Right plot: $A=0.4$, $D=-0.018$, $B=0.37$ and $C=E=0.49$.}}
  \label{fig:phase:noca}
\end{figure}

In Fig.~\ref{fig:phase:noca} we show the results of two such stroboscopic
sections.
As can be seen the left plot corresponds to completely integrable motions,
whereas in the right one sizeable chaotic zones coexisting with
stability islands, this strongly suggesting the applicability of the
KAM scenario.
However, our vector field is neither Hamiltonian nor time-reversible,
and then the KAM theory does not directly apply to this case.
However, we will show how a suitable change of variables can be performed
that unveils a time-reversible symmetry existing in our vector field.
For this purpose, we first recall that the structure of gradient vector
fields is preserved under orthogonal transformations.
In this way, if we consider the transformation $r=M s$, with $M^T = M^{-1}$,
applied to $\dot{r}=\nabla S(r,t)$,
we have that $\dot{s}=\nabla \tilde S(s, t)$, being $\tilde S(s,t)=S(Ms,t)$.
In other words, any orthogonal transformation can be performed
on the wave function instead of on the vector field.

\begin{lemma}\label{lem:cano}
If Eq.~\eqref{eq:wave:ini} satisfies the non-degeneracy
condition $BC\neq EF$,
then there exist an orthogonal transformation and a time shift,
such that the wave function takes the form
\begin{equation}\label{eq:wave:lem2}
\psi = \bigg(  \frac{\hat A \rme^{-\rmi t}}{\sqrt{\pi}}  +
\frac{2 x \hat B \rme^{-2\rmi t}}{\sqrt{2\pi}}  +
\frac{ 2 y \rmi \hat C\rme^{-2 \rmi t}}{\sqrt{2\pi}}  \bigg)
\rme^{-\frac{1}{2}(x^2+y^2)},
\end{equation}
where $\hat A, \hat B, \hat C \in \br$, $\hat B >0$, $\hat C \neq 0$
satisfy $\hat A^2 + \hat B^2 +\hat C^2 = 1$.
\end{lemma}

We will refer to the wave function~\eqref{eq:wave:lem2} as the canonical
form of~\eqref{eq:wave:ini}, and the rest of the paper is devoted to the
study of this case.
For this reason, the hat in the coefficients will be omitted,
since it is understood that $D=E=F=0$.
In Table~\ref{tab:prop} we give the actual values the canonical
coefficients after the transformation corresponding to the results
in Fig.~\ref{fig:phase:noca}.

\begin{proof}
For convenience, we consider the complexified phase space
$z=x + \rmi y$, so that the wave function~\eqref{eq:wave:ini}
results in
$$
\psi(z,\bar z,t)= \bigg(\hat \cala \rme^{-\rmi t} + \hat \calb
\rme^{-2 \rmi t} z + \hat \calc \rme^{-2\rmi t} \bar z \bigg)
\rme^{-\frac{1}{2} z \bar z},
$$
where $\sqrt{\pi} \hat \cala = A + \rmi D$, $\sqrt{2\pi} \hat \calb = B + C +\rmi (E-F)$,
and $\sqrt{2\pi} \hat \calc =B - C +\rmi (E+F)$.
Then, it is easy to check that the vortex, i.e. the set of points where
the wave function vanishes, has the following position with respect to time
$$
z_v(t) = \frac{-|\hat \cala| |\hat
\calb|\rme^{\rmi(t-b+a)} + |\hat \cala||\hat
\calc|\rme^{-\rmi(t-c+a)}}{|\hat \calb|^2-|\hat
\calc|^2},
$$
where $\hat \cala = |\hat \cala|\rme^{\rmi a}$, $\hat \calb = |\hat
\calb|\rme^{\rmi b}$ and $\hat \calc = |\hat \calc| \rme^{\rmi c}$.
Notice that the vortex is well defined thanks to the non-degeneracy assumption,
and its trajectory\footnote{We use here the term trajectory to refer to the
evolution of the vortex, despite the fact that it is not a solution of the ODE.}
follows an ellipse.
This ellipse does not appear in the usual canonical form,
but this can be made so by performing the rotation: $z \mapsto z \rme^{\rmi \mu}$
and the time shift: $t \mapsto t + \lambda$.
In this way
$$
z_v(t) = \frac{-|\hat \cala| |\hat
\calb|\rme^{\rmi(t-b+a-\mu+\lambda)} + |\hat \cala||\hat
\calc|\rme^{-\rmi(t-c+a+\mu+\lambda)}}{|\hat \calb|^2-|\hat
\calc|^2},
$$
where it is clear that by choosing $2\mu = c-b$ and $2\lambda = c+b-2a$,
the desired result is obtained.
Then, the corresponding wave function in these new coordinates is
$$
\psi = \bigg(|\hat \cala| \rme^{-\rmi t} + |\hat \calb| \rme^{-2\rmi t}
z + |\hat \calc| \rme^{-2\rmi t} \bar z \bigg)
\rme^{-\frac{1}{2} z \bar z} \rme^{\rmi (2a -
\frac{b+c}{2})},
$$
that can be further simplified since the factor $\rme^{\rmi(2a - \frac{b+c}{2})}$
plays no role in the Bohmian equations for the quantum trajectories.
Finally, by recovering the coefficients in cartesian coordinates, one obtains
$\hat A = \sqrt{\pi}|\hat \cala|$, $\hat B= \sqrt{\frac{\pi}{2}} (|\hat \calb|+
|\hat \calc|)$, $\quad \hat C= \sqrt{\frac{\pi}{2}} (|\hat \calb|-|\hat \calc|)$
and $\hat D=\hat E=\hat F=0$, which renders Eq.~\eqref{eq:wave:lem2}.
\end{proof}

\begin{table}[!t]
\centering
{\footnotesize
\begin{tabular}{ccc}
  \hline
    & Left plot & Right plot \\
  \hline
  $A$ & 0.370540146272978 & 0.400404795176082 \\
  $B$ & 0.656772411113622 & 0.705788460189184 \\
  $C$ & 0.656772411113622 & 0.584413081188110 \\
  \hline
\end{tabular}}
\caption{{\footnotesize Wave function coefficients in the canonical model
  corresponding to the results of Fig.~\protect\ref{fig:phase:noca}.}
  Hats have been omitted as discussed in the text.}
 \label{tab:prop}
\end{table}

\section{Study of the canonical form}\label{sec:cano}

Throughout the rest of the paper we consider the wave
function~\eqref{eq:wave:ini} with $A , C \neq0$, $B>0$ and $D=E=F=0$.
Let us remark that by changing the time $t \mapsto -t$, if necessary,
we can further restrict the study to the case $C>0$.
The corresponding quantum trajectories are then obtained from the
vector field
\begin{equation}\label{eq:case1}
X_\psi= \bigg( \frac{-2 BC y -\sqrt{2}A B \sin t}{V(x,y,t)},
\frac{2 B C x + \sqrt{2}A C \cos t }{V(x,y,t)}\bigg),
\end{equation}
where $V(x,y,t)= 2 B^2x^2 + 2 C^2 y^2 + 2 \sqrt{2} A B x \cos t + 2\sqrt{2}
A C y \sin t + A^2$.
In these coordinates, the only vortex of the system follows the
trajectory given by
$$
(x_v(t),y_v(t))=\bigg( -\frac{A}{\sqrt{2} B} \cos t,
-\frac{A}{\sqrt{2} C} \sin t \bigg),
$$
which corresponds to an ellipse of semi-axes $a = A/(\sqrt{2}B)$
and $b = A/(\sqrt{2}C)$, respectively.
%
\begin{figure}[!t] \centering
\includegraphics[width=6cm,height=6cm,angle=-90]{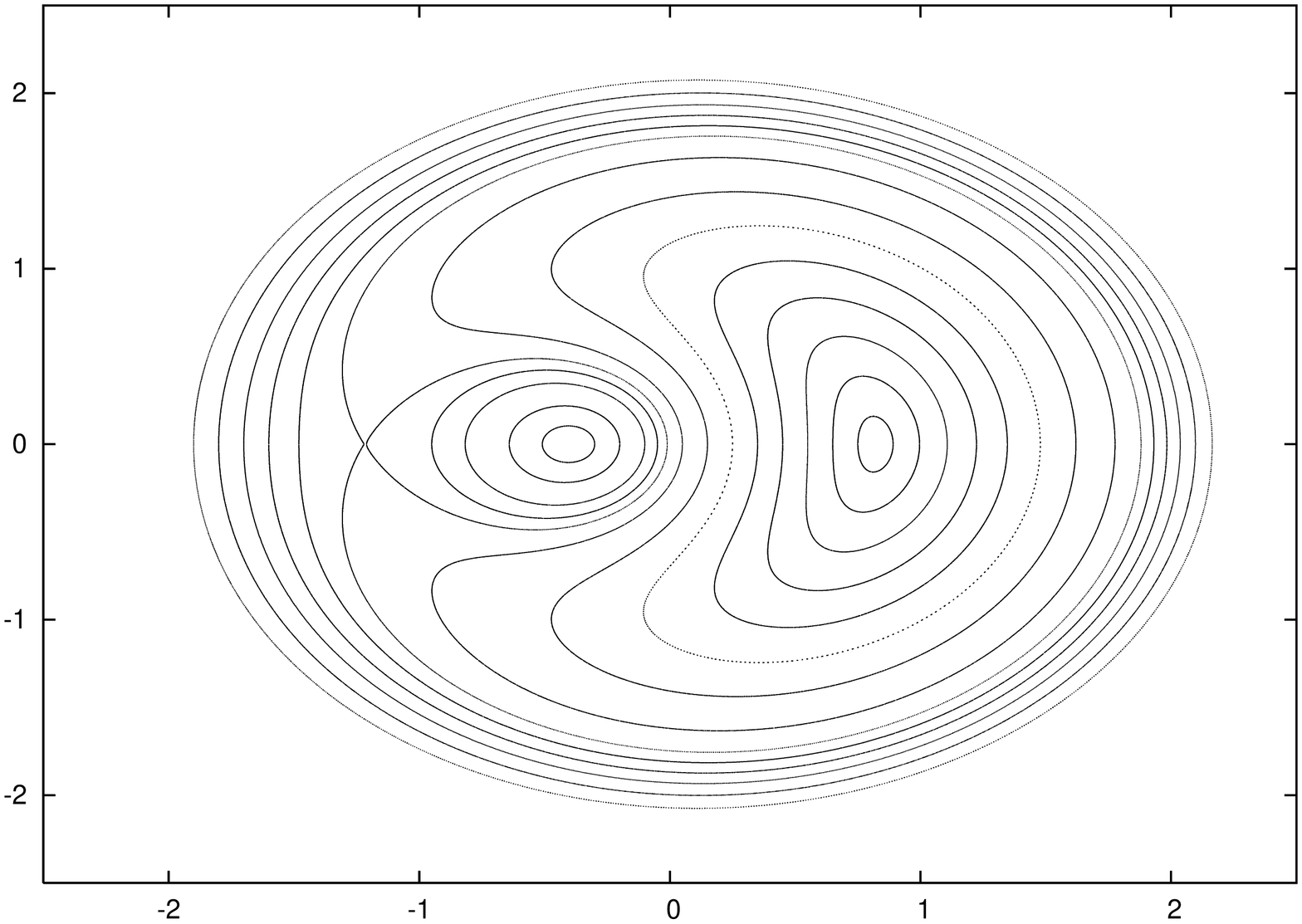}
\includegraphics[width=6cm,height=6cm,angle=-90]{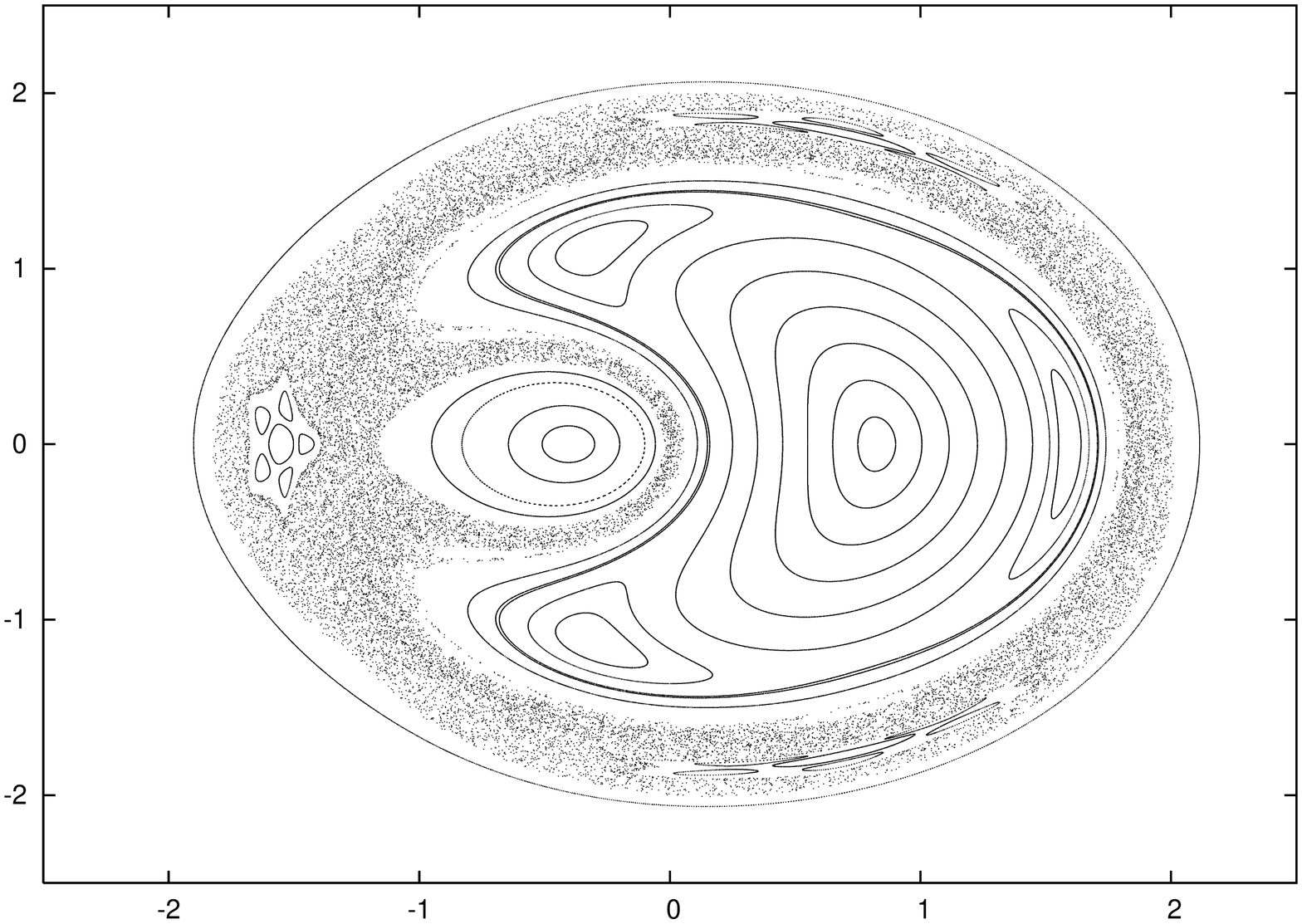}
\includegraphics[width=6cm,height=6cm,angle=-90]{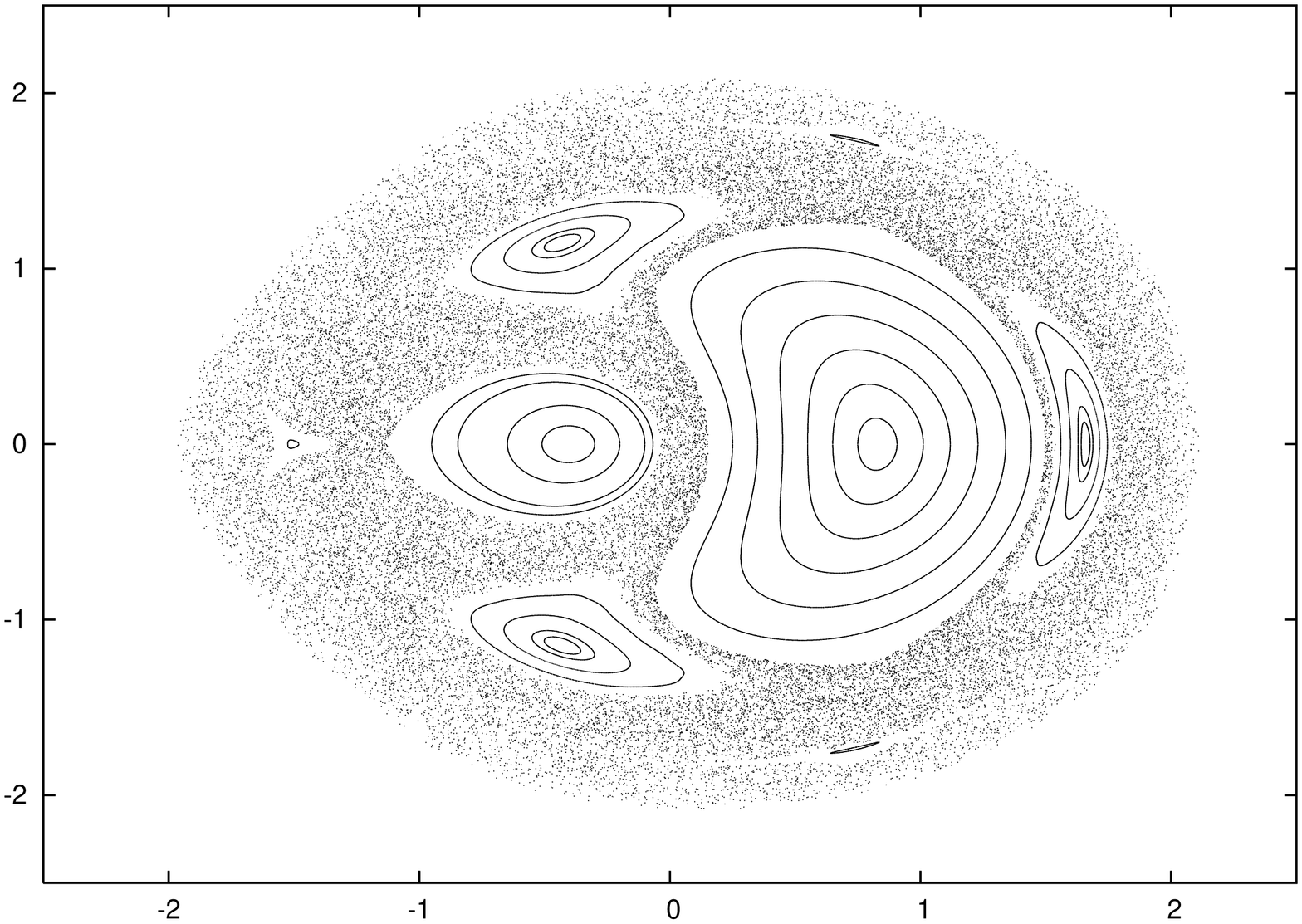}
\includegraphics[width=6cm,height=6cm,angle=-90]{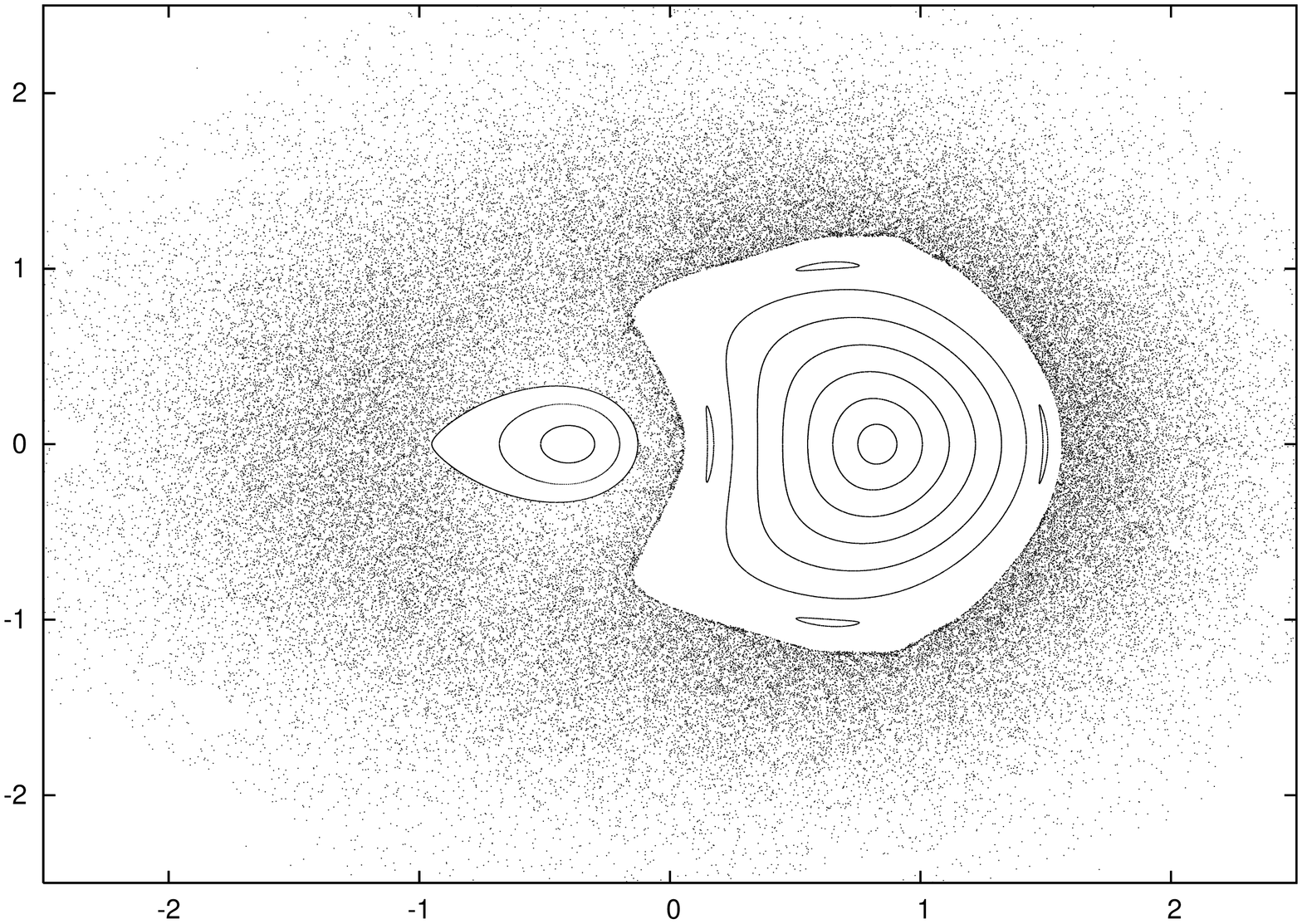}
\caption{{\footnotesize Stroboscopic $2\pi$-periodic sections
  corresponding to the quantum trajectories generated by the
  canonical velocity field defined by Eq.~(\ref{eq:case1}) for
  $a=0.4$ and, $b=0.4, 0.44, 0.48$ and $0.68$ from left-top to right-bottom,
  respectively.}} \label{fig:phase1}
\end{figure}

In Fig.~\ref{fig:phase1} we show some stroboscopic sections corresponding
to this (canonical) velocity field for different values of the
parameters $a$ and $b$.
As can be seen, a wide variety of dynamical behaviors,
characteristics of a system with mixed dynamics, is found.
In the left-top panel, which corresponds to the case in which $a=b$
(vortex moving in a circle),
we have sections corresponding to a totally integrable case.
As we move from left to right and top to bottom some of these tori
are broken, and these areas of stochasticity coexist with others in
which the motion is regular, this including different chains of
islands.
Moreover, the size of the chaotic regions grows as the value of
$b$ separates from that of $a$.

This variety of results can be well understood and rationalize by
using some standard techniques of the field of dynamical systems,
in the following way.
Although the vector field~\eqref{eq:case1} is not Hamiltonian,
it is time-reversible with respect to the involution $\Theta(x,y)=(x,-y)$.
This result is very important for the purpose of the present paper,
since it implies that the KAM theory applies to our system
if we are able to write down our vector field in the form
$X_\psi = X_0 + \ep X_1$, $\ep \ll 1$,
being the dynamics corresponding to $X_0$ integrable and $X_1$
time-reversible.
More specifically, let us assume that $X_0$ does not depend on $t$ and
$X_1$ be $2\pi$-periodic with respect to $t$.
Moreover, let us assume that for $X_0$ there exists a family of periodic orbits
whose frequency varies along the family (non-degeneracy condition).
Then, our result guarantees that when the effect of the perturbation $\ep X_1$
is considered, most of the previous periodic orbits give rise to invariant tori
of frequencies $(1,\omega)$, where $\omega$ is the frequency of the unperturbed
periodic orbit.
Of course, the persistence of these objects is conditioned to the fact that
the vector $(1,\omega)$ satisfies certain arithmetic conditions
(see~\cite{Broer96,Sevryuk07} for details).
Since these arithmetic conditions are fulfilled for a big (in the sense of
the Lebesgue measure) set of the initial orbits, the important hypothesis
that we have to check in order to ascertain the applicability of the KAM
theory is the non-degeneracy of the frequency map.

In our problem, two such integrable cases exists.
First, if $A=0$ the vortex is still at the origin and the time periodic
part in the vector field disappears.
As a consequence, all the quantum orbits of the system appear as ellipses
centered at the origin in the $xy$-plane.
It will be shown in the next section that the corresponding
frequency varies monotonically along the orbits.
This case has not been explicitly included in Fig.~\ref{fig:phase1}
due to its simplicity.
Second, and as will be analyzed in Sect.~\ref{ssec:nonautonomous},
if $B=C$, or equivalently $a=b$, that is the vortex moves in a circle,
the vector field is also integrable for any value of $A$.
The corresponding stroboscopic sections are shown in the top-left
panel of Fig.~\ref{fig:phase1}).
Here, the structure of the phase space changes noticeably,
since two new periodic orbits, one stable and the other unstable,
appear.
Moreover, the obtained integrable vector field depends on $t$.
We will show that this time dependence can be eliminated by means of a
suitable change of coordinates, showing that our problem remains in the
context described in the previous paragraph.
The rest of the panels in Fig.~\ref{fig:phase1} can be understood as
the evolution of this structure as the perturbation,
here represented by the difference between $B$ and $C$,
as dictated by the KAM theorem.

To conclude the paper, let us now discuss in detail the two integrable
cases in the next two sections.

\section{The integrable autonomous case}\label{ssec:autonomous}

For $A=0$, $B\neq 0$ and $C\neq 0$ it is easily seen that the vector
field~\eqref{eq:case1} is integrable.
Actually, the orbits of the quantum trajectories in the $xy$-plane are
ellipses around the origin (position at which the vortex is fixed).
Also, the frequency of the corresponding trajectories approaches infinity
as they get closer to the vortex position.
Let us now compute the frequency $\omega$ of these solutions.
First, we introduce a new time variable $\tau$,
satisfying $dt/d\tau = B^2 x^2+ C^2y^2$, and then solve the
resulting system, thus obtaining
\begin{equation}\label{eq:xy:tau}
x(\tau)  =  \alpha \cos (BC \tau + \beta ), \quad  y(\tau)  =
\alpha \sin (BC \tau + \beta),
\end{equation}
where $\alpha$ is the distance from the vortex.
Next, we recover the original time, $t$, by solving the differential equation
defining the previous change of variables
\begin{align*}
\frac{dt}{d\tau} = \alpha^2 \frac{B^2+C^2}{2}+\alpha^2 \frac{B^2-C^2}{2} \cos (2BC \tau +
 2\beta),
\end{align*}
whose solution is given by
$$
\underbrace{\frac{2}{\alpha^2(B^2+C^2)}}_{\gamma} t = \tau +
\underbrace{\frac{B^2-C^2}{2(B^2+C^2)BC}}_\delta \sin (2BC\tau +
2\beta).
$$
Notice that this equation is invertible since $|2BC\delta| < 1$,
and then $\tau = \gamma t + f(2BC \gamma t)$, $f$ being a
$2\pi$-periodic function.
Finally, introducing this expression into~\eqref{eq:xy:tau},
one can conclude that the solution has a frequency given by
$$
\omega = BC \gamma =\frac{2BC}{\alpha^2(B^2+C^2)}
$$
that varies monotonically with respect to the distance to the vortex.
Then, for $A \ll 1$, the existence of invariant tori around
the vortex is guaranteed.

\section{The integrable non-autonomous case}
  \label{ssec:nonautonomous}

Let us consider now the case of non-vanishing values of
$B=C$ for any $A\neq 0$.
In this case, we have $a=b\neq 0$, and system~\eqref{eq:case1}
can be written as
\begin{equation}\label{eq:case1:BeqC}
X_\psi(x,y,t)= \bigg( \frac{-y -a \sin(t)}{\tilde V(x,y,t)},
\frac{x +a \cos(t) }{\tilde V(x,y,t)}\bigg),
\end{equation}
where $\tilde V(x,y,t)= (x+a \cos(t))^2+(y+a \sin(t))^2$.
This vector field corresponds to the following
Hamiltonian
$$
H(x,y,t)= - \frac{1}{2} \ln \tilde V(x,y,t),
$$
that it is actually integrable.

\begin{lemma}\label{lem:integra}
Let us consider $e$ (energy), the symplectic variable conjugate
to $t$, and define the autonomous Hamiltonian
$\calh_1 (x,y,t,e)=H(x,y,t)+e$, then we have that
$$
\calh_2(x,y,t)= \tilde V(x,y,t) \rme^{-x^2-y^2}
$$
is a first integral of $\calh_1$, in involution and functionally
independent.
As a consequence, if $a=b\neq 0$ the system is completely integrable.
\end{lemma}

\begin{proof} It is straightforward to see that the Poisson bracket
with respect to the canonical form $dx \wedge dy + dt \wedge de$ satisfies
$\{\calh_1, \calh_2 \}=0$.
Moreover, $\calh_2$ does not depend on $e$,
so that it is an independent first integral.
\end{proof}

Taking these results into account, one can completely understand
the picture presented in the top-left plot of Fig.~\ref{fig:phase1}.
Since the system is integrable, it is foliated by invariant tori,
despite the two periodic orbits that
are created by a resonance introduced when parameter $A$ changes from
$A=0$ to $A\neq 0$.
Next, we characterize these two periodic orbits:

\begin{lemma}\label{lem:periodic}
If $A> 0$, the system has two periodic orbits given by
$$
r_\pm(t)=(x_\pm(t),y_\pm(t))= ( a_\pm \cos t, a_\pm \sin t),
$$
where the coefficients $a_+$ and $a_-$ are given by
$a_\pm = \frac{-a \pm \sqrt{a^2+4}}{2}$.
Moreover, the orbit $r_-(t)$ is hyperbolic with characteristic exponents
$ \pm (a_-^4-1)^{1/2}$, and $r_{+}(t)$ is elliptic with
characteristic exponents $ \pm \rmi (1-a_+^4)^{1/2}$.
If $A<0$ the same result holds just switching the roles of $a_+$ and $a_-$.
\end{lemma}

\begin{proof}
It is known that if the sets $\{\calh_2^{-1}(c), \, c\in \br\}$
are bounded differentiable submanifolds,
their connected components carry quasi-periodic dynamics.
Moreover, the critical points of $\calh_2$ determine the periodic
orbits of the system.
Therefore, these periodic orbits are given by expressions:
$2 x \tilde V=\frac{\partial \tilde V}{\partial x}$,
and $2 y \tilde V=\frac{\partial \tilde V}{\partial y}$,
which can also be written as
\begin{eqnarray*}
x\,((x+a\cos t)^2+(y+a\sin t)^2) & =  x+a\cos t,\\
y\,((x+a\cos t)^2+(y+a\sin t)^2) & =  y+a\sin t.
\end{eqnarray*}
from which we obtain our two periodic orbits: $x_\pm(t)=a_\pm \cos t$,
and $y_\pm(t)= a_\pm \sin t$, with $a_\pm = 1/(a_\pm + a)$. In addition, it is
easy to check that $a^2_->1$ and $a^2_+<1$, respectively.

Finally, the stability of these orbits can be obtained by
considering the following associated variational equations,
\begin{equation}
  \label{eq:varia}
\left(
  \begin{array}{c}
    \dot w_1 \\
    \dot w_2 \\
  \end{array}
\right) = a^2_\pm \left(
  \begin{array}{cc}
    \sin (2 t) & -\cos (2 t) \\
    -\cos (2 t) & - \sin (2 t)\\
  \end{array}
\right) \left(
  \begin{array}{c}
     w_1 \\
     w_2 \\
  \end{array}
\right),
\end{equation}
Solutions for this equation can be easily obtained by using the complex variable
$z=w_1+\rmi w_2$, and solving $\dot{z}=-\rmi \rme^{2 \rmi t} a^2_\pm \bar{z}$.
We have the following set of fundamental solutions
\begin{align*}
w_1(t) & =\rme^{\pm t \sqrt{a^4_--1}} \bigg( (1-a^2_-) \cos t \mp
\sqrt{a^4_--1} \sin t \bigg),\\
w_2(t) & =\rme^{\pm t \sqrt{a^4_--1}} \bigg( (1-a^2_-) \sin t \pm
\sqrt{a^4_--1} \cos t \bigg),
\end{align*}
for the hyperbolic case, and
\begin{align*}
w_1(t)={} &  \cos (\pm t \sqrt{1-a^4_+} )(\pm \sqrt{1-a^4_+} \cos t - (1+a^2_+) \sin t) \\
       & + \sin (\pm t \sqrt{1-a^4_+} )(\pm \sqrt{1-a^4_+} \cos t + (1+a^2_+) \sin t),\\
w_2(t)={} &   \cos (\pm t \sqrt{1-a^4_+})(\pm \sqrt{1-a^4_+} \sin t + (1+a^2_+) \cos t) \\
       & + \sin (\pm t \sqrt{1-a^4_+})(\pm \sqrt{1-a^4_+} \sin t - (1+a^2_+) \cos t),
\end{align*}
for the elliptic one.
Finally, the corresponding characteristic exponents can be obtained by
a straightforward computation of the monodromy matrix.
\end{proof}

\begin{remark}
Notice that the chaotic sea observed in Fig.~\ref{fig:phase1} is associated to
the intersection of the invariant manifolds of the hyperbolic periodic orbit
that we have computed.
\end{remark}

Now, and in order to apply KAM theorem, we compute locally the frequency map
of this unperturbed system around the vortex and the elliptic periodic orbit.
To this end, we perform a symplectic change of coordinates in a neighborhood
of these objects in order to obtain action-angle variables up to third order
in the action.

In general, let $\calh(x,y,t,e)=H(x,y,t)+e$ be a Hamiltonian that is $2\pi$-periodic
with respect to $t$ and has a first integral, $F(x,y,t)$.
Let us consider the generating function, $\tilde S(x,t,I,E) = t E + S(x,t,I)$,
determining a symplectic change of variables $(x,y,t,e) \mapsto (I,\theta,t, E)$
defined implicitly by
\[
y=\frac{\partial S}{\partial x}, \qquad e=E+\frac{\partial S}{\partial t},
   \qquad \theta=\frac{\partial S}{\partial I}
\]
where $\theta$ is also $2\pi$-periodic.
This transformation is introduced in such a way that the new Hamiltonian
depends only on $I$
\begin{equation}
  \label{eq:hamint}
H \left( x,\frac{\partial S}{\partial x}, t \right) +
   \frac{\partial S}{\partial t} = h(I).
\end{equation}
Since a first integral of the system is known, we can define the corresponding
action as
\begin{equation}
  \label{eq:action}
I=F(x,y,t)=F\left( x,\frac{\partial S}{\partial x}, t \right).
\end{equation}
From Eq.~\eqref{eq:action}, we obtain locally the equation
$\frac{\partial S}{\partial x} = f(x,t,I)$,
so that we have $S(x,t,I)= \int f(x,t,I) dx +g(t,I)$.
Introducing this expression into~\eqref{eq:hamint},
we obtain the following equation for $g$:
\begin{equation}
  \label{eq:partialg}
\frac{\partial g}{\partial t} = h(I)-H(x,f,t)-\int \frac{\partial f}{\partial t}dx
\end{equation}
and can conclude that, since $g$ must be $2\pi$-periodic with respect to $t$,
then $h(I)$ has to satisfy
\begin{equation}
  \label{eq:ham:int}
h(I) = \left\langle H(x,f,t) + \int \frac{\partial f}{\partial t}dx \right\rangle =
\left\langle H(x,f,t) \right\rangle,
\end{equation}
where $\langle \cdot \rangle$ denotes average with respect to $t$.
Finally, we notice that since $F$ is a first integral, we can define $g$
so that $\theta$ becomes $2\pi$-periodic.

Computations are simplified observing that the left hand side in Eq.~\eqref{eq:partialg}
does not depend on $x$ (we use the fact that $F$ is a first integral),
so we can set $x=0$.
According to this, we have to solve $\hat F(\hat f,t)=I$,
where $\hat F(\cdot,t)=F(0,\cdot, t)$ and then we
have to compute the average $h(I)= \langle \hat H(\hat f,t) \rangle$,
where $\hat H(\cdot,t)= H(0, \cdot,t)$.

First, let us consider a neighborhood of the vortex for $A>0$.
To this end, we introduce the new variables $x=-a \cos t+\Delta_x$
and $y=-a \sin t +\Delta_y$,
so that the Hamiltonian $H_v=H_v(\Delta_x,\Delta_y,t)$ and the first integral
$F_v=F_v(\Delta_x,\Delta_y,t)$ are
\begin{align}
H_v & = -\frac{1}{2} \ln (\Delta_x^2+\Delta_y^2) - a \Delta_x \cos t -
  a \Delta_y \sin t, \label{eq:hamv}\\
F_v & =(\Delta_x^2+\Delta_y^2)\,\exp (-a^2+2a \Delta_x \cos t +
  2 a \Delta_y \sin t -\Delta_x^2 - \Delta_y^2 ),
 \nonumber
\end{align}

\begin{proposition}
  \label{prop:vor}
There exist a symplectic change of variables
$(\Delta_x,\Delta_y,t,e) \mapsto (I,\theta,t, E)$, with $\theta\in \bt$,
setting the vortex at $I=0$, such that the new Hamiltonian becomes
$$
h_v(I) = -\frac{1}{2} \ln I - \frac{a^2}{2}- \frac{\rme^{a^2} I}{2} -
 \frac{3 a^2 \rme^{2a^2}}{2} I^2 + \calo_3(I).
$$
\end{proposition}

\begin{proof}
According to the above discussion,
we have $\hat F_v(\hat f_v,t)=\hat f_v^2 \rme^{-a^2+ 2 a \hat f_v \sin t
  - \hat f_v^2}=I$.
Then, by introducing this expression in~\eqref{eq:hamv} ones obtains
$$
\hat H_v(\hat f_v,t)= -\frac{1}{2} \ln I -\frac{a^2}{2} -\frac{\hat f_v^2}{2}.
$$
Finally, we only have to compute the first terms in the expansion of $\hat f_v^2$
obtaining
\[
\hat f_v^2 = \rme^{a^2}I-2a\rme^{\frac{3}{2}a^2} \sin t \, I^{3/2}
  + 6a^2 \rme^{2a^2} \sin^2 t \,I^2 + \ldots
\]
and use that $\langle \sin t \rangle =0$ and $\langle \sin^2 t \rangle = \frac{1}{2}$.
\end{proof}

On the other hand, a neighborhood of the elliptic periodic orbit for $A>0$
can be studied by means of the variables
$x= a_+ \cos t+\Delta_x$ and $y=a_+ \sin t +\Delta_y$.
One thus obtains that the Hamiltonian and the first integral are given by
\begin{align*}
H_+ & = -\frac{1}{2} \ln V_+ + a_+ \Delta_x \cos t+ a_+ \Delta_y \sin t,\\
F_+ & = V_+ \exp (-a_+^2 - 2a_+ \Delta_x \cos t - 2 a_+ \Delta_y \sin t -\Delta_x^2
  - \Delta_y^2),
\end{align*}
where $V_+ = (a+a_+)^2 (1+2a_+ \Delta_x \cos t + 2 a_+ \Delta_y \sin t +
a_+^2 \Delta_x^2 + a_+^2 \Delta_y^2)$.

\begin{proposition}\label{prop:op}
There exist a symplectic change of variables
$(\Delta_x,\Delta_y,t,e) \mapsto (I,\theta,t, E)$,
with $\theta \in \bt$, setting the periodic orbit at $I=(a_++a)^2 \rme^{-a^2}$,
such that the new Hamiltonian becomes
$$
h_+(I) = - \ln (a+a_+)^2 + \frac{1-\Pi_1}{2}J - \frac{1+2\Pi_2}{4} J^2+ \calo_3(I),
$$
where we have introduced the notation
$$J=1-\frac{\rme^{a^2}I}{(a_++a)^{2}}$$
and also
$$\Pi_1 = \frac{1}{\sqrt{1-a_+^4}}, \qquad \Pi_2 =
  \frac{a^2(41 a^8 -88 a^6 + 119 a^4-54 a^2+18)}{36 \sqrt{1-a^4} (a^8+1-2 a^4)(1+a^2)}.$$
\end{proposition}

\begin{proof}
As before, we consider a solution $\hat f_+(I,t)$ for the equation $\hat F_+(\hat f_+,t)=I$.
For convenience, we introduce the notation $I=(a_++a)^2 \rme^{-a^2}(1-J)$
in order to set the periodic orbit at $J=0$.
Then, it turns out that the expression
\[
(1-a_+^2 \cos (2 t)) \hat f^2_+ + \bigg(a_+(1+a^2) \sin t - \frac{8}{3} a^3_
  + \sin^3 t \bigg) \hat f^3_+ + \calo_4(\hat f_+) = J,
\]
approximates the previous equation for $\hat f_+$ and that the following
expansion in terms of $J$
$$
\hat f_+^2= \alpha_1(t) J + \alpha_{3/2}(t) J^{3/2} + \alpha_2(t) J^2 + \ldots
$$
holds, where
\begin{align*}
\al_1(t)     = & \frac{1}{1-a^2_+ \cos 2t},\\
\al_{3/2}(t) = & \frac{-a_+(1+a^2_+) \sin t + \frac{8}{3} a^3_+ \sin^3 t}{(1-a^2_+
  \cos(2t))^{5/2}}, \\
\al_2(t)     = & \frac{3}{2} \frac{(a_+(1+a_+)^2\sin t-\frac{8}{3} a^3_+
  \sin^3 t)^2}{(1-a^2_+ \cos 2t)^4}.
\end{align*}

Hence, we have to compute the average of
$$
\hat H_+(\hat f_+,t)=- \ln (a_++a)^2 - \frac{1}{2} \ln (1-J) - \frac{\hat f^2_+}{2}
$$
that follows from the fact that $\langle \alpha_1 \rangle =
\Pi_1$, $\langle \alpha_{3/2} \rangle=0$ and $\langle \alpha_2 \rangle = \Pi_2$.
These averages are computed easily by using the method of residues.
\end{proof}

\section{Conclusion}

In this paper we present an scheme to study in a
systematic way the intrinsic stochasticity and general complexity
of the quantum trajectories that are the basis of quantum mechanics
in the formalism developed by Bohm in the 1950's.
In our opinion this approach, which based on the ideas and results
of the dynamical systems theory, can seriously contribute to
establish firm grounds that foster the importance of the conclusions
of future studies relying on such trajectories, thus avoiding errors
and ambiguities that has happened in the past.
As an illustration we have considered the simplest, non--trivial
combination of eigenstates of the two dimensional isotropic harmonic
oscillator.

The corresponding velocity field is put in a so--called canonical
form, and the characteristics of the corresponding quantum
trajectories studied in detail.
It is proved that only one vortex and two periodic orbits, one elliptic
and the other hyperbolic, organize the full dynamics of the system.
In it, there exist invariant tori associated to the vortex and the elliptic
periodic orbit.
Moreover, there is a chaotic sea associated to the hyperbolic periodic orbit.
The KAM theory has been applied to this scenario by resorting to a suitable
time-reversible symmetry, that is directly observed in the canonical form
for the velocity field determining the quantum trajectories of the system.
It should be remarked that the results reported here concerning
the hyperbolic periodic orbit constitute a generalization of those
previously reported in~\cite{WisniackiP05}, in the sense that here
a more concise and constructive approach to the associated dynamics,
is presented.
We summarize the dynamical characteristics of the different
possibilities arising from the canonical velocity field \eqref{eq:case1}
in Table~\ref{table:roadmap}, that represents a true road--map to navigate
across the dynamical system, i.e.\ quantum trajectories, that are defined
based on the pilot effect \cite{Bohm} of the wave function \eqref{eq:wave:ini}.
Also, note that the generic model, i.e.~when $E,F$ or $G$ do not vanish,
does not satisfy any of the properties considered in the table.
%
\begin{table}[!t]
\centering
{\footnotesize
\begin{tabular}{lcccl}
  \hline
     &            &             & Time        & Stroboscopic \\
     & Integrable & Hamiltonian & Reversible  & sections     \\

  \hline
  $A=0, B\ne 0, C\ne 0$        & yes &  no & yes & Ellipses \\
                               &     &     &     & around origin \\
  $A \ne 0, B=C$               & yes & yes & yes & Top-left panel \\
                               &     &     &     & in Fig.~\ref{fig:phase1} \\
  $A \ne 0, B \ne C$           &  no &  no & yes & Rest of panels    \\
                               &     &     &     & in Fig.~\ref{fig:phase1} \\
  \hline
\end{tabular}}
\caption{{\footnotesize Dynamical characteristics of the quantum
  trajectories generated  from the different
  possibilities in the canonical model \eqref{eq:case1} for the
  pilot wave function \eqref{eq:wave:ini}.}}
 \label{table:roadmap}
\end{table}

Finally, the method presented here is, in principle, generalizable
to other more complicated situations in which more vortex and
effective dimensions exist.
Some methods have been described in the literature that can be
applied to these situations \cite{Jordi}.
This will be the subject of future work.

\section{Acknowledgments}

This work has been supported by the Ministerio de Educaci\'on y
Ciencia (Spain) under projects FPU AP2005-2950, MTM2006-00478,
MTM2006--15533 and CONSOLIDER 2006--32 (i--Math), and the
Comunidad de Madrid under project S--0505/ESP--0158 (SIMUMAT).
The authors gratefully acknowledge useful discussions with Carles Sim\'o,
and to Gemma Huguet for encouragement.
A.L.\  also thanks the hospitality of the Departamento de Qu\'{\i}mica  at UAM
during different stays along the development of this work.

\bibliographystyle{plain}

\end{document}